%% file: ms.tex
\documentclass[sigconf,authorversion]{acmart}

\usepackage{booktabs} % For formal tables
\usepackage[all]{xy}
\usepackage{hhline}
\usepackage{subcaption}
\usepackage{rotating}
\usepackage{multirow}
\usepackage{amsmath,bussproofs,epsfig,graphicx,verbatim,bm,galois,xcolor,xspace,todonotes}

\usepackage{microtype}
\usepackage{listings}
\usepackage{tikz}
\usetikzlibrary{arrows,shapes,positioning,automata}
\usepackage[ruled,vlined,linesnumbered]{algorithm2e}

\newcommand*{\Scale}[2][4]{\scalebox{#1}{$#2$}}

\newcommand{\sbr}[1]{\lbrack \! \lbrack #1 \rbrack \! \rbrack}

\newcommand{\FeatExp}{\textit{FeatExp}}

\newcommand{\AP}{{\textit{AP}}\xspace}

\newcommand{\impskip}[1][]{\ensuremath{\mbox{\texttt{skip}}^{#1}}}
\newcommand{\impassign}[3][]{\ensuremath{{#2} ~\mbox{\texttt{:=}}^{#1}~ {#3}}}

\newcommand{\true}{\textrm{true}}
\newcommand{\false}{\textrm{false}}

\newcommand{\Kk}{\ensuremath{\mathbb{K}}}
\newcommand{\Ff}{\ensuremath{\mathbb{F}}}

\newcommand{\Zz}{\ensuremath{\mathbb{Z}}}

\newcommand{\Fff}{\mathcal{F}}
\newcommand{\Ttt}{\ensuremath{\mathcal{T}}}
\newcommand{\trans}{\textit{trans}\xspace}

\newcommand{\joinasym}{\ensuremath{\bm{\alpha}^{\textnormal{\textrm{join}}}}}
\newcommand{\joingsym}{\ensuremath{\bm{\gamma}^{\textnormal{\textrm{join}}}}}

\newcommand{\SIMPLE}{{\small \textsc{SIMPLE}}}

\newcommand{\SPIN}{\textsc{SPIN}}
\newcommand{\snip}{\ensuremath{\overline{\text{SNIP}}}}
\newcommand{\proveline}{\ensuremath{\text{ProVeLine}}}
\newcommand{\Promela}{\textsc{Promela}}

\newcommand{\ft}{\textcolor{red}{\ensuremath{d}}}

\newcommand{\fc}{\textcolor{brown}{\ensuremath{c}}}

\newcommand{\Sketch}{\textsc{Sketch}}
\newcommand{\SPLSketch}{\textsc{PromelaSketcher}}
\newcommand{\Loop}{{\small \textsc{Loop}}}
\newcommand{\Welfare}{{\small \textsc{Welfare}}}
\newcommand{\LoopCond}{{\small \textsc{LoopCond}}}
\newcommand{\Salesman}{{\small \textsc{Salesman}}}

\DeclareMathOperator*{\Max}{Max}
\DeclareMathOperator*{\Min}{Min}

\tikzset{
  treenode/.style = {align=center, inner sep=0pt, text centered,
    font=\sffamily},
  arn_n/.style = {treenode, circle, %white, font=\sffamily,  %fill=black,%\bfseries,
    draw=black, text width=2em},% arbre rouge noir, noeud noir
  arn_r/.style = {treenode, circle, red, draw=red,
    text width=2em, very thick},% arbre rouge noir, noeud rouge
  arn_x/.style = {treenode, rectangle, draw=black,
    minimum width=1.25em, minimum height=1.25em}% arbre rouge noir, nil
}

\acmConference[SAC '22]{The 37th ACM/SIGAPP Symposium on Applied Computing}{April 25--29, 2022}{Virtual Event, }
\acmDOI{10.1145/3477314.3507170}
\begin{document}
\title{Model Sketching by Abstraction Refinement for Lifted Model Checking (Extended Version)}
\titlenote{This is an extended version of the paper that was published at The 37th ACM/SIGAPP Symposium on Applied Computing (SAC '22).}

\renewcommand{\shorttitle}{Model Sketching by Abstraction Refinement for Lifted Model Checking}

\author{Aleksandar S. Dimovski}
\orcid{0000-0002-3601-2631}
\affiliation{%
  \institution{Faculty of Informatics, Mother Teresa University}
  \streetaddress{Mirche Acev, nr. 4}
  \city{Skopje}
  \country{North Macedonia}
  \postcode{1000}
}
\email{aleksandar.dimovski@unt.edu.mk}

\renewcommand{\shortauthors}{A.S.~Dimovski}

\begin{abstract}
In this work, we show how the use of verification and analysis techniques for model families (software product lines) with numerical features
provides an interesting technique to synthesize complete models from sketches (i.e.\ partial models with holes).
In particular, we present an approach for synthesizing \Promela\, model sketches using variability-specific
abstraction refinement for \emph{lifted (family-based) model checking}.
\end{abstract}




\maketitle

\section{Introduction}\label{sec:introduction}

\input{introduction.tex}

%\section{Motivating Example}\label{sec:motivate}
%\input{motivate.tex}

\section{Model Families} \label{sec:families}

\input{families.tex}

\section{Syntactic Transformations} \label{sec:sketches}

\input{sketches.tex}

\section{Synthesis Algorithm} \label{sec:synthesis}
\input{synthesis.tex}

\section{Evaluation} \label{sec:evaluation}

\input{evaluation.tex}

%\section{Related Work}
%\input{related.tex}

\section{Conclusion}

\input{conclusion.tex}

\bibliographystyle{ACM-Reference-Format}
\bibliography{ms}

\appendix

\section{Proofs} \label{app:proofs}
\input{proofs.tex}

\input{bench.tex}

\end{document}

%% file: introduction.tex
This paper presents a novel synthesis framework for reactive models that adhere to a given set of properties.
The input is a \emph{sketch} \cite{DBLP:journals/sttt/Solar-Lezama13}, i.e.\ a partial model
with holes, where each  hole is a placeholder that can be replaced with
one of finitely many options; and a \emph{set of properties} that the model needs to fulfill.
Model sketches are represented in the \Promela\, modelling language \cite{spin} and properties are expressed in LTL \cite{katoen-beier}.
The synthesizer aims to generate as output a \emph{sketch realization}, i.e.\ a complete model instantiation, which satisfies the given properties by suitably filling the holes.
%Model sketching is a natural instance of model synthesis \cite{DBLP:conf/popl/PnueliR89,DBLP:conf/fmcad/AlurBJMRSSSTU13,DBLP:journals/sttt/Solar-Lezama13,DBLP:conf/tacas/KatzP08,DBLP:journals/corr/KatzP14}.
%We consider the possible deletion and/or addition of transitions and atomic propositions in states.
%The domains such as protocols and concurrency present an especially appealing target for sketching,
%where models has to be small and accurate but some algorithmic details are required.

%Checking all realisations one by one guarantees that given enough time, a correct solution will be found.
%However, sketching is an intractable problem in nature, which becomes quite quickly prohibitively expensive.
%Realisation explosion can make it virtually impossible to complete model sketches with very large realisation spaces.

In this work, we frame the model sketching problem as a verification/analysis problem for
model families (a.k.a.\, Software Product Lines -- SPLs) \cite{model-checking-spls}, and then formulate
an abstraction refinement algorithm that operates on model families to efficiently solve it.
SPL methods and architectures allow building a family of similar models, known
as \emph{variants} (\emph{family members}), from a common code base. % by maximizing reuse in order to decrease development cost and time-to-market.
A custom variant is specified in terms of suitable \emph{features} %(statically configured options)
selected for that particular variant at compile-time.
%We consider model families implemented using $\texttt{\#if}$ directives from the C preprocessor \texttt{CPP} \cite{SPL,DBLP:journals/ese/HunsenZSKL0A16}.
%An $\texttt{\#if}$ directive specifies under which conditions parts of code should be included or excluded from a variant.
%The SPL method has grown in popularity over the last 20 years,
%especially in the domain of embedded systems (e.g., cars, phones, avionics), where organizing development in
%product lines is very common.

All possible model sketch realizations constitute a model family,
where each hole is represented by a numerical feature with the same domain.
In contrast to Boolean features that have only two values,
numerical features can have a range of numbers as explicit values.
Hence, the model sketching problem reduces to selecting correct variants (family members) from the resulting model family \cite{DAL}.
The automated analysis of such families for finding a correct variant is challenging since in addition to the state-space explosion
affecting each family member, the family size (i.e., the number of variants) typically grows exponentially in the number of features.
%This is particulary apparent in the case of model families that contain numerical features with big domains,
%thus admitting astronomic family sizes.
%This also affects the model sketching problem leading to the so-called realization space explosion.
A naive \emph{brute force enumerative} solution is to check each individual
variant of the model family by applying an off-the-shelf model checker.
This is shown to be very inefficient for large families \cite{model-checking-spls,DBLP:conf/aosd/MidtgaardBW14}.

This paper applies an abstraction refinement procedure over the compact, all-in-one, representation of model families, called featured transition system (FTS) \cite{model-checking-spls,spin15,sttt16}, to solve the model sketching problem.
More specifically, we first devise variability abstractions tailored for model families that contain numerical features.
%This way, we extend the previous variability abstractions \cite{spin15,sttt16} that were designed for
%model families with only Boolean features.
Variability abstractions represent a configuration-space reduction technique that
compresses the entire model family (with many configurations and variants) into an abstract model (with a single abstract configuration and variant),
so that the result of model checking
a set of LTL properties in the abstract model is preserved in all variants of the model family.
\begin{comment}
The variability abstractions forget in which variant the abstract model operates.
In particular, the abstraction aggregates multiple variants in a single abstract model,
which is then fed to an off-the-shelf model checker.
If no counterexample is found in the abstract model, then all variants satisfy the given property and
any of them represents a solution to the sketching problem.
Otherwise, the counterexamples are analysed and classified as either \emph{genuine}, which correspond
to execution behaviours of some concrete variants, or \emph{spurious}, which are introduced due to the abstraction.
If a genuine counterexample exist, the corresponding variants do not satisfy the given properties; otherwise the abstraction is too coarse and
a spurious counterexample is used to refine the abstract model.
The procedure is then repeated on the refined abstract models that represent suitable sub-families of the original model family.
The abstraction and refinement are done in an efficient manner
as source-to-source transformations of \Promela\, code, which makes our procedure easy to implement/maintain as a simple
meta-algorithm script.
This way, we also avoid the need for intermediate storage in memory of the
 concrete full-blown models.
This approach often drastically reduces the average time and search steps required for finding
one correct variant (i.e., solution to the sketching problem), compared to the brute-force enumerative
approach. % and specialized lifted model checkers.
\end{comment}
The procedure is first applied on an abstract model that represents the entire model family, and then is repeated on refined abstract models that represent suitable sub-families of the original model family.
Hence, the abstraction refinement approach \cite{spin15,sttt16,DBLP:conf/fase/DimovskiLW19,DBLP:journals/sttt/Dimovski20,DBLP:journals/tcs/DimovskiLW20} starts from considering
all possible variants, and successively splits the entire family into indecisive
and incorrect sub-families with respect to the given set of properties.
The approach is sound and complete: either a correct complete model (variant) does exist and it is computed,
or no such model exists and the procedure reports this.
Because of its special structure and possibilities for sharing of equivalent execution behaviours and model checking results
for many variants,
this algorithm is often able to converge to a solution very fast after a handful of iterations
even for sketches with large search spaces.

We have implemented our prototype model synthesizer, called \SPLSketch. % built on top of the \SPIN\, model checker \cite{spin}.
It uses variability-specific abstraction refinement for lifted model checking
of model families with numerical features, and calls the \SPIN\, model checker \cite{spin} to verify the generated abstract models.
The abstraction and refinement are done in an efficient manner
as source-to-source transformations of \Promela\, code, which makes our procedure easy to implement/maintain as a simple
meta-algorithm script.
We illustrate this approach for
automatic completion of various \Promela\, model sketches. % with very large realization spaces.
We also compare its performance with the brute-force %enumerative
approach.

%% file: families.tex
\paragraph{Featured transition system.}
Let $\Ff = \{A_1, \ldots, A_k\}$ be a finite and totally ordered set of \emph{numerical features} available
in a model family.
Let $\mathrm{dom}(A) \subseteq \Zz $  denote the set of possible values that can be assigned to feature $A$.
 % For any numerical feature $A \in \Ff_{\Zz}$, we have $\mathrm{dom}(A) \subseteq \Zz$.
%Note that any Boolean feature can be represented as a numerical feature $B \in \Ff$ with $\mathrm{dom}(B)=\{0,1\}$.
% such that 0 means that feature $B$ is disabled while 1 means that $B$ is enabled.
A valid combination of feature's values represents a \emph{configuration} $k$, which specifies
one \emph{variant} of a model family.
It is given as a \emph{valuation function} $k:\Ff \to \Zz$, which assigns a value from $\mathrm{dom}(A)$
to each feature $A$. % i.e.\ $k(A) \in \mathrm{dom}(A)$ for any $A \in \Ff$.
We assume that only a subset \(\Kk \) of all possible configurations are \emph{valid}.
Each configuration $k \in \Kk$ can be given by a formula:
$(A_1\!=\!k(A_1)) \land \ldots \land (A_k\!=\!k(A_k))$.
%We often abbreviate $(B\!=\!1)$ with $B$ and $(B\!=\!0)$ with $\neg B$, for a Boolean feature $B \in \Ff$.
%The set of valid configurations $\Kk$ can be also represented as a formula: $\lor_{k \in \Kk} k$.
%We will use both representations interchangeably.

%We use \emph{transition systems} (TS) to describe behaviours of single systems.
A \emph{transition system} \cite{katoen-beier} is a tuple $\Ttt=(S,I,\trans,AP,L)$,
which is used to describe behaviours of single systems.
We write $s_1 \longrightarrow s_2$ whenever \((s_1,s_2) \in \trans\).
A \emph{path} of a TS $\Ttt$ is an \emph{infinite} sequence $\rho = s_0 s_1 s_2 \ldots$
with $s_0 \in I$ s.t. $s_i \stackrel{}{\longrightarrow} s_{i+1}$ for all $i \geq 0$.
The \emph{semantics} of a TS $\Ttt$, denoted $\sbr{\Ttt}_{TS}$, is the set of its paths.

A \emph{featured transition system} (FTS) represents a compact model,
which describes the behaviour of a whole family of systems %in a \emph{superimposed} manner.  This means that it combines models of many variants
in a single monolithic description.
%Their transitions are guarded by a \emph{presence condition} that identifies the variants they belong to.  Presence conditions $\psi$ are from
The set of feature expressions, $\FeatExp(\Ff)$, are propositional logic formulae over constraints of $\Ff$:  %generated by the grammar:
$\psi ::= \true \, |  \, A \bowtie n \, | \, \neg \psi \, | \, \psi \land \psi$,
where $A \in \Ff$, $n \in \Zz$, and $\bowtie \,\in \{=, <\}$.
We write $\sbr{\psi}$ %to denote the set of configurations from $\Kk$ that satisfy $\psi$, i.e.\ $k \in \sbr{\psi}$ iff $k \models \psi$.
for the set of configurations that satisfy $\psi$, i.e.\ $k \in \sbr{\psi}$ iff $k \models \psi$.
A \emph{featured transition system} (FTS) is
$\Fff\!\!=\!(S,\!I,\!\trans,AP,L,\Ff,\Kk,\delta)$, where $(S, I, \trans, AP, L)$
form a TS; $\Ff$ is a set of available features; $\Kk$ is a set of valid configurations; and
$\delta: \trans \!\to\! \FeatExp(\Ff)$ is a total function decorating transitions with presence conditions (feature expressions).
The \emph{projection} of an FTS $\Fff$ to a configuration $k \in \Kk$, denoted as $\pi_k(\Fff)$, is the TS
$(S,I,\trans',AP,L)$, where $\trans'=\{ t \in \trans \mid k \models \delta(t) \}$.
We lift the definition of \emph{projection} to sets of configurations \(\Kk' \!\subseteq\! \Kk\),
 denoted as $\pi_{\Kk'}(\Fff)$, by keeping transitions admitted by at least one of configurations in $\Kk'$.
That is, $\pi_{\Kk'}(\Fff)$, is the FTS
$(S,I,\trans',AP,L,\Ff,\Kk',\delta')$, where $\trans'=\{ t \in \trans \mid \exists k \in \Kk'. k \models \delta(t) \}$
and $\delta'$ is the restriction of $\delta$ to $\trans'$.
The \emph{semantics} of an FTS $\Fff$, denoted as $\sbr{\Fff}_{FTS}$, is the union of paths %(behaviours)
of the projections on all valid variants $k \in \Kk$, i.e.\ $\sbr{\Fff}_{FTS} = \cup_{k \in \Kk} \sbr{\pi_k(\Fff)}_{TS}$.

\paragraph{Abstraction.}
We start working with Galois connections
between Boolean complete lattices of feature expressions, and then induce a notion of abstraction of FTSs.
The Boolean complete lattice of feature expressions is: \((\FeatExp(\Ff)_{/\equiv},\models,\lor,\land,\true,\false,\neg)\), where
the elements of
$\FeatExp(\Ff)_{/\equiv}$ are equivalence classes of formulae \(\psi\) obtained by quotienting by the semantic equivalence $\equiv$.
%The ordering $\models$ is the standard entailment between propositional logics formulae,
%whereas the least upper bound and the greatest lower bound are just logical disjunction and conjunction respectively.
%Finally, the constant \false\ is the least, \true\ is the greatest element, and negation is the complement operator.

The \emph{join abstraction}, $\joinasym_{\Kk} : \FeatExp(\Ff) \to\FeatExp(\emptyset)$, replaces each feature expression $\psi$ in an FTS with \true\ if there exists at least one configuration from $\Kk$ that satisfies $\psi$.
The abstract sets of features and configurations are: $\joinasym_{\Kk}(\Ff)=\emptyset$ and $\joinasym_{\Kk}(\Kk) = \{ \true \}$.
% if $\Kk \neq \emptyset$.  %For the empty set of configurations, the abstraction is not defined.
The abstraction and concretization functions
between $\FeatExp(\Ff)$ and $\FeatExp(\emptyset)$, which form a Galois connection \cite{sttt16}, are:
\[
\begin{array}{@{}ll@{}}
\Scale[0.85]{\joinasym_{\Kk}(\psi) \!=\!
  \begin{cases}
    \true  & \textrm{if } \exists k \in \Kk. k \models \psi \\
    \false & \textrm{otherwise}
  \end{cases}}, &
\Scale[0.8]{\joingsym_{\Kk}(\true) \!=\! \true, \joingsym_{\Kk}(\false) \!=\! \bigvee_{k \in 2^\Ff\!\setminus\! \Kk} k}
%  \begin{cases}
%  \true & \text{if } \psi \text{ is } \true\\
%  \bigvee_{k \in 2^\Ff\!\setminus\! \Kk} k & \text{if } \psi \text{ is } \false
%  \end{cases}
\end{array}
\]
%The concretization of \false\,  returns the complement of the
%valid configuration set: $2^\Ff\!\setminus\! \Kk$.
%In this way, we obtain a single abstract model (variant) that includes all transitions occurring in some concrete variants.

Given the FTS $\Fff=(S,I,\trans,AP,L,\Ff,\Kk,\delta)$, %$[\chi] \phi$ be an fLTL formula,
%and $(\alpha,\gamma)$ be a Galois connection.
we will define a TS $\joinasym_{\Kk}(\Fff)=(S,I,\trans',AP,L)$
to be its \emph{abstraction},
where $\trans' \!=\! \{ t \!\in\! \trans \mid \joinasym_{\Kk}(\delta(t)) \!=\!  \true \}$.
Note that transitions in the abstract TS $\joinasym_{\Kk}(\Fff)$ describe the behaviour that is possible in some variants of the concrete FTS $\Fff$,
but not need be realized in the other variants.
The information about which transitions are associated with which variants is lost, thus
causing a precision loss in the abstract model.
This way,
$\sbr{\joinasym_{\Kk}(\Fff)}_{TS} \supseteq \cup_{k \in \Kk} \sbr{\pi_k(\Fff)}_{TS}$.
We say that a TS $\Ttt$ satisfies a LTL formula $\phi$, written $\Ttt \models \phi$, iff all paths of $\Ttt$ satisfy formula $\phi$ \cite{katoen-beier}.
We say that an FTS $\Fff$ satisfies $\phi$, written $\Fff \models \phi$, iff all its variants satisfy $\phi$, i.e.\
$\forall k\!\in\!\Kk. \, \pi_k(\Fff) \models \phi$.
%Since a TS satisfies an LTL formula if all its paths satisfy it,
%the following holds.

\begin{theorem}[Preservation results, \cite{sttt16}] \label{theorem:sound}
For every $\phi \in LTL$ \cite{katoen-beier}, $\joinasym_{\Kk}(\Fff) \models \phi \, \implies \, \Fff \models \phi$.
\end{theorem}

The problem of evaluating $\Fff \models \phi$ can be reduced to a number of smaller problems by
partitioning the configuration space $\Kk$.
Let the subsets $\Kk_1, \Kk_2, \ldots, \Kk_n$ form a \emph{partition} of $\Kk$.
Then, $\Fff \models \Phi$ iff $\pi_{\Kk_i}(\Fff) \models \phi$ for all $i=1,\!\ldots\!,n$.
\begin{comment}
\begin{corollary}[\cite{sttt16}]
Let $\Kk_1, \Kk_2, \ldots, \Kk_n$ form a \emph{partition} of $\Kk$.
%and $(\alpha_1,\!\gamma_1), \ldots, (\alpha_n,\!\gamma_n)$ be Galois connections.
If $\joinasym_{\Kk_1}(\pi_{\Kk_1}(\Fff)) \models \phi \, \land \ldots \land \, \joinasym_{\Kk_n}(\pi_{\Kk_n}(\Fff)) \models \phi$,
 then $\Fff \models \phi$. Moreover, if $\joinasym_{\Kk_j}(\pi_{\Kk_j}(\Fff)) \models \phi$
 then $\pi_{k}(\Fff) \models \phi$ for all $k \in \Kk_j$.
\end{corollary}
%The above result shows that we can partition the family of variants \Kk\, into sub-families
%$\Kk_1, \ldots, \Kk_n$, such that if the abstract model of some sub-family satisfies a property
%then all variants from that sub-family also satisfy the property.
\end{comment}

\paragraph{Abstraction Refinement Framework.}
The abstraction refinement procedure \texttt{ARP} for checking $\Fff \models \phi$
is illustrated by Algorithm~\ref{Algorithm1}.
We first construct an initial abstract model $\joinasym_{\Kk}(\Fff)$,
and check $\joinasym_{\Kk}(\Fff) \models \phi$ (Line 1).
If the abstract model satisfies the given property (i.e., the counterexample $c$ is \texttt{null}), then all variants from \Kk\, satisfy it and we stop.
In this case, the global variable \texttt{end} is also set to \true\, making all other recursive calls to \texttt{ARP} to end (Lines 2, 6, 10).
 Otherwise, a non-\texttt{null} counterexample $c$ is found.
Let $\psi$ be the feature expression computed by conjoining feature expressions labelling all transitions that
belong to path $c$ when $c$ is simulated in $\Fff$ (Line 3).
%This simulation of path $c$ from $\joinasym_{\Kk}(\Fff)$ in $\Fff$ can be done due to the fact that $\joinasym_{\Kk}(\Fff)$ and $\Fff$
%have the same control structures (i.e., states and transitions), except that transitions in $\Fff$
%are guarded by feature expressions.
There are two cases to consider.

\begin{algorithm}[t]
\SetKwInOut{Global}{Global}
\KwIn{An FTS $\Fff$, a configuration set $\Kk$, and an LTL formula $\phi$}
\KwOut{Correct variants $k \in \Kk$, s.t. $\pi_k(\Fff) \models \phi$}
\Global{$\texttt{end:=}\false$}
$c = (\joinasym_{\Kk}(\Fff) \models \phi)$ \;
\lIf{($c\texttt{=null}$)}  {
\{$\texttt{end:=}\true$; \bf{return} $\Kk$ \} }
$\psi \texttt{:=FeatExp(c)}$ \;
\eIf{($\texttt{sat}(\psi \!\land\! (\bigvee_{k \in \Kk}k))$)}  {
$(\psi_1,\ldots,\psi_n) \texttt{:= Split}(\sbr{\neg \psi} \cap \Kk)$ \;
\lIf{(\texttt{end})}{\bf{return} $\emptyset$}
$ARP(\pi_{\sbr{\psi_1}}(\Fff),\sbr{\psi_1},\phi); \ldots; ARP(\pi_{\sbr{\psi_n}}(\Fff),\sbr{\psi_n},\phi)$
}
{
$\psi' = \texttt{CraigInterpolation}(\psi,\Kk)$\;
\lIf{(\texttt{end})}{\bf{return} $\emptyset$}
$ARP(\pi_{\sbr{\psi'}}(\Fff),\sbr{\psi'},\phi); \, ARP(\pi_{\sbr{\neg \psi'}}(\Fff),\sbr{\neg \psi'},\phi)$
}

\caption{{\bf \texttt{ARP($\Fff,\Kk,\phi$)}} \label{Algorithm1}}
\end{algorithm}

First, if $\psi \land (\bigvee_{k \in \Kk}k)$ is satisfiable (i.e.\ $\Kk \cap \sbr{\psi} \neq \emptyset$),
then the found counterexample $c$ is \emph{genuine} for
variants in $\Kk \cap \sbr{\psi}$. For the other variants from $\Kk \cap \sbr{\neg \psi}$, the found counterexample cannot
be executed (Lines 5,6,7). We call \texttt{Split} to split
the space $\Kk \cap \sbr{\neg \psi}$ in sub-families $\sbr{\psi_1}, \ldots, \sbr{\psi_n}$, % depending on the structure of $\psi$.
such that all atomic constraints in $\psi_i$ are of the form: $A \bowtie n$, where $A \in \Ff$ and $n \in dom(A)$.
In particular, the  \texttt{Split} function takes as input a set of configurations and returns a list of sets of configurations.
For example, assume that we have two numerical features $\Min \!\leq\! \texttt{A} \!\leq\! \Max$ and $\Min \!\leq\! \texttt{B} \!\leq\! \Max$.
If $\psi=(\texttt{A}\!=\!3)$,
then \texttt{Split}($\sbr{\neg \psi}$) is $(\Min \!\leq\! \texttt{A} \!\leq\! 2) \land (\Min \!\leq\! \texttt{B} \!\leq\! \Max)$ and $(4 \!\leq\! \texttt{A} \!\leq\! \Max) \land (\Min \!\leq\! \texttt{B} \!\leq\! \Max)$.
%if $\psi=(\texttt{A}\!=\!3) \land (\texttt{B}=2)$, then \texttt{Split}($\sbr{\neg \psi}$) is $(\Min \!\leq\! \texttt{A} \!\leq\! 2) \!\land\! (\Min \!\leq\! \texttt{B} \!\leq\! 1)$, $(\Min \!\leq\! \texttt{A} \!\leq\! 2) \!\land\! (3 \!\leq\! \texttt{B}\!\leq\! \Max)$, $(4 \!\leq\! \texttt{A}\!\leq\! \Max) \!\land\! (\Min \!\leq\! \texttt{B} \!\leq\! 1)$ and $(4 \!\leq\! \texttt{A}\!\leq\! \Max) \!\land\! (3 \!\leq\! \texttt{B}\!\leq\! \Max)$.
%but if $\psi=(\texttt{A=}\Min)$, then \texttt{Split}($\sbr{\neg \psi}$) is $(\Min \!<\! \texttt{A} \!\leq\! \Max) \land (\Min \!\leq\! \texttt{B} \!\leq\! \Max)$.
Finally, we call \texttt{ARP} to verify the sub-families:
$\pi_{\sbr{\psi_1}}(\Fff), \ldots, \pi_{\sbr{\psi_n}}(\Fff)$. 
Note that if $\Kk \cap \sbr{\neg \psi} = \emptyset$, \texttt{Split} updates variable \texttt{end} to \true\ 
and so no recursive \texttt{ARP}s are called.

Second, if $\psi \land (\bigvee_{k \in \Kk}k)$ is unsatisfiable (i.e.\ $\Kk \cap \sbr{\psi} = \emptyset$), then the found counterexample $c$ is \emph{spurious} for all variants in $\Kk$ (due to incompatible feature expressions) (Lines 9,10,11).
A feature expression $\psi'$ used for constructing refined sub-families is determined
by means of Craig interpolation \cite{DBLP:conf/tacas/McMillan05} from $\psi$ and $\Kk$.
First, we find the minimal unsatisfiable core $\psi^c=X \land Y=\false$ of $\psi \!\land\! (\bigvee_{k \in \Kk}k)$.
Next, the interpolant $\psi'$ is computed, such that
$\psi'$ summarizes and translates why $X$ is inconsistent with $Y$ in their shared language.
Finally,
we call the \texttt{ARP} to check $\pi_{\sbr{\psi'}}(\Fff) \models \phi$ and
$\pi_{\sbr{\neg \psi'}}(\Fff) \models \phi$.
By construction, it is guaranteed that the spurious counterexample $c$ does not occur in both $\pi_{\sbr{\psi'}}(\Fff)$
and $\pi_{\sbr{\neg \psi'}}(\Fff)$ \cite{fase17}.

Note that abstract models we obtain are ordinary TSs where all feature expressions %associated with transitions of $\Fff$
 are replaced with $\true$. Therefore, the verification step $\joinasym_{\Kk}(\Fff) \models \phi?$ (Line 1) can be performed using a single-system model checker such as \SPIN.
Also note that we call \texttt{ARP} until we find a correct variant (variable \texttt{end} is set to \true) or the updated set of configurations $\Kk$
becomes empty. Therefore,
\texttt{ARP($\Fff,\Kk,\phi$)} terminates and is correct.

%% file: sketches.tex
We now present the high-level modelling language \Promela\, for writing sketches and model families. % and show its TS semantics.
Then, we describe several transformations of \Promela\, sketches and model families.
% as well as how to construct abstract models and projections out of model families.

%\subsection{\Promela\, language}

\paragraph{Syntax of \Promela.}
\Promela\ \cite{spin} is a non-deterministic modelling language designed for
describing systems composed of concurrent processes that communicate
asynchronously.
%A \Promela\ model, $P$, consists of a finite set of processes to be executed concurrently.
The basic statements of
processes are given by:
$$
\begin{array}{@{}l}
stm ::=  \impskip \!\mid\! \texttt{break} \!\mid\! \impassign{\texttt{x}\!}{\!expr} \!\mid\! c \texttt{?} x \!\mid\! c \texttt{!} expr \!\mid\! stm_1\texttt{;} stm_2 \!\mid\!  \\
\ \!\texttt{if}\!::\!g_1 \!\rightarrow\! stm_1\cdots\!::\!g_n \!\rightarrow\! stm_n \, \texttt{fi} \!\mid\! \texttt{do}\!::\!g_1 \!\rightarrow\! stm_1\cdots\!::\!g_n \!\!\rightarrow\!\! stm_n \texttt{od}
\end{array}
$$
where \texttt{x} is a variable, $expr$ is an expression, $c$ is a channel, and $g_i$ are conditions
over variables and contents of channels.
\begin{comment}
The ``$\texttt{if}$'' is a non-deterministic choice between
statements $stm_i$ for which the guard $g_i$ evaluates to $true$ for
the current evaluation of variables. If none of the guards $g_1,
\ldots, g_n$ are $true$ in the current state, then the
statement blocks.  Similarly, the
``$\texttt{do}$'' loop represents an iterative execution of the
non-deterministic choice among statements $stm_i$ for which the
guard $g_i$ holds in the current state.
If guards $g_1,
\ldots, g_n$ do not hold in the current state then the
loop blocks, whereas if \texttt{break} is executed then the
loop terminates. Note that the guard `\true' can be omitted in ``$\texttt{if}$''-s and ``$\texttt{do}$''-s.
Statements are preceded by a
declarative part, where variables and channels are declared.
\end{comment}

\paragraph{Sketches.}
To encode sketches, a single sketching construct of type expression  is included: a basic integer hole denoted by \texttt{??}.
Each hole occurrence is assumed to be uniquely labelled as $\texttt{??}_i$, and it has a bounded integer domain $[n_i,n'_i]$.
%We will sometimes write $\texttt{??}_{i}^{[n,n']}$ to make explicit the domain of a hole.
%The aim of the synthesizer is to replace each $\texttt{??}_{i}^{[n,n']}$ with a suitable integer constant from $[n,n']$,
%so that the resulting model will avoid any property failures.

\paragraph{Model Families.}
To encode multiple variants,
a new compile-time guarded-by-features statement is included:
\[
stm ::=  \ldots \mid \texttt{\#if}:: \psi_1 \rightarrow stm_1 \ldots ::\psi_n \rightarrow stm_n \ \texttt{\#endif}
\]
where $\psi_1, \ldots, \psi_n$ are
feature expressions defined over $\Ff$.
%We write $\Bb(\Ff)$ for the set of all feature (Boolean) expressions over $\Ff$.
The ``$\texttt{\#if}$'' statement contains feature expressions $\psi_i \in \FeatExp(\Ff)$ as presence conditions
(guards).
%Actually, this is the only place where features may be used.
If presence condition $\psi_i$ is satisfied by a configuration $k \in \Kk$ the statement $stm_i$
will be included in the variant corresponding to $k$.
Hence, ``\texttt{\#if}'' plays the same role as
``$\texttt{\#if}$'' directives in C preprocessor CPP~\cite{DBLP:conf/oopsla/KastnerGREOB11,DBLP:conf/fase/DimovskiAL21,DBLP:journals/scp/DimovskiAL22}.
%Note that guards $\psi_1, \ldots, \psi_n$ in an ``$\texttt{\#if}$'' statement are usually
%mutually disjoint, so only one of them can be satisfied by a configuration $k \in \Kk$.
%The hole expression \texttt{??} is not allowed in families.
\begin{comment}
The variant (model) corresponding to the configuration $k \in \Kk$ is obtained by a
projection function $\mathcal{P}_k$, which is an identity for basic statements and recursively
pre-processes all sub-statements of compound statements. For ``$\texttt{\#if}$'',
we have:
\[
\mathcal{P}_{k}(\texttt{\#if}:: \psi_1 \rightarrow \overline{stm_1} \!\ldots\! ::\psi_n \rightarrow \overline{stm_n} \ \texttt{\#endif}) \!=\!
    \begin{cases}
        \mathcal{P}_{k}( \overline{stm_i} )  & \!\textrm{if} \, k \models \psi_i  \\
        \impskip & \!\textrm{if} \, k \not\models (\psi_1 \!\lor\! \ldots \!\lor\! \psi_n)
  \end{cases}
\]
\end{comment}
The semantics of \Promela\ models and \Promela\ model families are given in \cite{spin,sttt16}.

\paragraph{Syntactic Transformations.}
Our aim is to transform an input sketch $\hat{P}$ with a set of $m$ holes $\texttt{??}_{1}^{[n_1,n_1']}, \ldots, \texttt{??}_{m}^{[n_m,n_m']}$,
into an output model family $\overline{P}$ with a set of numerical features $A_1, \ldots, A_m$ with domains $[n_1,n_1'], \ldots, [n_m,n_m']$. The set of configurations $\Kk$ includes all possible combinations of feature's values.
The rewrite rule for eliminating holes \texttt{??} from a model sketch is of the form:
\begin{equation} \label{transf.1}
\begin{array}{@{}l@{}}
\Scale[0.9]{ \!stm[\texttt{??}^{[n,n']}] ~\leadsto~
 \texttt{\#if}\!::\! \texttt{(A=}n) \rightarrow stm[n] \ldots \!::\!\texttt{(A=}n') \rightarrow stm[n'] \, \texttt{\#endif} } \tag{R-1}
\end{array}
\end{equation}
where $stm[-]$ is a (non-compound) basic statement with a single expression $-$ in it, $\texttt{??}^{[n,n']}$ is an occurrence of a hole with domain $[n,n']$, and \texttt{A} is a fresh numerical feature with domain $[n,n']$.
The meaning of the rule (\ref{transf.1}) is that if the current sketch being transformed
matches the abstract syntax tree node of the shape $stm[\texttt{??}^{[n,n']}]$ then replace $stm[\texttt{??}^{[n,n']}]$
according to the rule (\ref{transf.1}).

We write $\texttt{Rewrite}(\hat{P})$ to be the final model family
obtained by repeatedly applying the rule (\ref{transf.1}) on sketch $\hat{P}$ and
on its transformed versions until we reach a point where it can not be applied. % i.e.\ when all occurrences of holes \texttt{??} in $\hat{P}$ are eliminated.

We now present the syntactic transformations of model families $\overline{P} = \texttt{Rewrite}(\hat{P})$ obtained
from \Promela\, sketches $\hat{P}$.
We consider two transformations: projection $\pi_{\sbr{\psi}}(\overline{P})$ and variability abstraction $\joinasym_{\Kk}(\overline{P})$.
Let $\overline{P}$ represent a model family. % for which the sets of features and valid configurations are $\Ff$ and $\Kk$.
%We denote with $\sbr{P}$ the FTS obtained for this system family.

The projection $\pi_{\sbr{\psi}}(\overline{P})$ is obtained by defining a translation
recursively over the structure of $\psi$.
Let $\psi$ be of the form ($\texttt{A} \!<\! m$).
The rewrite rule is of the form:
\begin{equation} \label{transf.2}
\begin{array}{@{}l@{}}
\Scale[0.86]{\texttt{\#if}\!::\!\! \texttt{(A=}n) \!\to\! stm[n] \!\ldots \!::\!\texttt{(A=}m) \!\to\! stm[m] \ldots \!::\!\texttt{(A=}n') \!\to\! stm[n'] \texttt{\#endif} \!\leadsto\! } \\
\Scale[0.86]{ \texttt{\#if}\!::\!\! \texttt{(A=}n) \!\to\! stm[n] \!\ldots \!::\! \false \!\to\! stm[m] \ldots \!::\! \false \to stm[n'] \texttt{\#endif} }
 \tag{{\small R-2}}
\end{array}
\end{equation}
where $n \leq m \leq n'$.
That is, all guards that do not satisfy ($\texttt{A} \!<\! m$) are replaced with \false.
%The rewrite rule for the case ($A \!=\! m$) is defined similarly.
Let $\psi$ be a feature expression of the form $\neg \psi'$.
We first transform $\overline{P}$ by applying the projection $\psi'$, %but we keep in a special memo list all guards that become \false\, and the exact places where they occur.
then in all \texttt{\#if}-s obtained from the projection $\psi'$ we change the guards:
those guards of the form $(A=m')$ become \false, and \false\, guards are returned to the form $(A=m')$ by looking at a special memo list where we keep record of them.
Let $\psi$ be %a feature expression
of the form $\psi_1 \land \psi_2$.
Then, we apply projections $\psi_1$ and $\psi_2$ one after the other.

The abstract model $\joinasym_{\Kk}(\overline{P})$ is obtained by appropriately
resolving all ``$\texttt{\#if}$''-s.
%Thus, $\joinasym_{\Kk}(\overline{P})$
%represents a standard \Promela\, model without \texttt{\#if}-s and holes \texttt{??}.
The rewrite rule is: % of the form:
\begin{equation} \label{transf.3}
\begin{array}{@{}l@{}}
\texttt{\#if}:: (\psi_1) \rightarrow stm_1 :: \ldots ::(\psi_n) \rightarrow stm_n \texttt{\#endif} \leadsto \\
\qquad \texttt{if}:: \joinasym_{\Kk}(\psi_1) \rightarrow stm_1 :: \ldots :: \joinasym_{\Kk}(\psi_n) \rightarrow stm_n \, \texttt{fi} \tag{R-3}
\end{array}
\end{equation}
where all guards in the new $\texttt{if}$ are set to $true$ or $false$ depending whether there is some valid configurations
that satisfies that guard. % or not.

The correctness of these transformations are formally proved by structural induction on $\hat{P}$ and $\overline{P}$
(see Theorems~\ref{thm:rewrite} and ~\ref{theo:transform} in App.~\ref{app:proofs}).

%% file: synthesis.tex
We can now encode the sketch synthesis problem as a lifted model checking problem.
In particular, we delegate the effort of conducting an effective search of
all possible sketch realizations to an efficient abstraction refinement for lifted model checking.
Once the lifted model checking of the corresponding model family is performed,
we can see for which variants the given property is correct.
Those variants represent the correct sketch
realizations. % that satisfy the property.

The synthesis algorithm $\textrm{SYNTHESIZE}(\hat{P},\phi)$ for solving a sketch $\hat{P}$ is the following.
The sketch $\hat{P}$ is first encoded as a model family $\overline{P} = \texttt{Rewrite}(\hat{P})$.
Then, we call function \texttt{ARP($\overline{P},\Kk,\phi$)},
which takes as input the model family $\overline{P}$, its configuration set $\Kk$, and the property to verify $\phi$,
and returns as solution a set of variants $\Kk' \subseteq \Kk$ that satisfy $\phi$ obtained
after performing the ARP. % abstraction refinement procedure given in Section~\ref{sec:ARP}.
The correctness and termination of $\textrm{SYNTHESIZE}(\hat{P},\phi)$ are shown in Theorem~\ref{theo:synth} in App.~\ref{app:proofs}.

%Finally, the inferred set of correct variants
%$\Kk' \subseteq \Kk$ is returned as solution.
%of the sketching problem.

\begin{comment}
\begin{algorithm}[t]
\KwIn{A sketch $\hat{P}$ and an LTL property $\phi$}
\KwOut{Correct realizations $\hat{P}$}
$\overline{P} = \texttt{Rewrite}(\hat{P})$\;
$\Kk' = \texttt{ARP}(\overline{P},\Kk,\phi)$\;
\bf{return} $\Kk'$\;
 %}
\caption{{\bf \texttt{SYNTHESIZE}$(\hat{P},\phi$)} \label{Algorithm10}}
\end{algorithm}

\begin{theorem} \label{theo:synth}
\texttt{SYNTHESIZE}$(\hat{P},\phi$) is correct and terminates.
\end{theorem}
\begin{proof}
The procedure
\texttt{SYNTHESIZE($\hat{P},\phi$)} terminates since all steps in it are terminating.
The correctness of \texttt{SYNTHESIZE($\hat{P},\phi$)} follows from the correctness of \texttt{Rewrite} (see Theorem~\ref{thm:rewrite}),
\texttt{ARP} (see Theorem~\ref{the:ARP}) and syntactic transformations (see Theorem~\ref{theo:transform}).
\end{proof}
\end{comment}

%% file: evaluation.tex
%We now present how several interesting numerical sketches are converted
%into full programs using our approach based on symbolic lifted analysis.

%We now evaluate our approach for model sketching by abstraction refinement for lifted model checking.
%The evaluation aims to show that the proposed synthesis algorithm can be used to successfully resolve several interesting
%model sketches, and that it outperforms the brute force enumeration approach.

\newcommand{\ra}[1]{\renewcommand{\arraystretch}{#1}}

\begin{table*}[t]
\ra{0.9}
\caption{Performance results. % of \SPLSketch\, vs. \textsc{Brute force}. % performed on selected benchmarks from SV-COMP 2019.
%We report results for different sizes of domains of \texttt{??} and input variables: 5-, 8-, and 16-bits.
All times in sec.
}\label{fig:performance1}
\vspace{-2mm}
\begin{tabular}{@{}l|cl|cl|cl|cl|cl|cl@{}}\toprule
\multirow{2}{*}{Bench.} & \multicolumn{4}{c|}{3 bits}  &  \multicolumn{4}{c|}{4 bits}   & \multicolumn{4}{c}{8 bits} \\
\cmidrule{2-5} \cmidrule{6-9} \cmidrule{10-13}
     & \multicolumn{2}{c|}{\scriptsize{\textsc{PromelaSketcher}}}  &  \multicolumn{2}{c|}{\footnotesize{\textsc{Brute-force}}} & \multicolumn{2}{c|}{\scriptsize{\textsc{PromelaSketcher}}}  & \multicolumn{2}{c|}{\footnotesize{\textsc{Brute-force}}}  & \multicolumn{2}{c|}{\scriptsize{\textsc{PromelaSketcher}}} & \multicolumn{2}{c}{\footnotesize{\textsc{Brute-force}}} \\
%     & \multicolumn{2}{c|}{\scriptsize{\textsc{Sketcher}}}  &  \multicolumn{2}{c|}{\footnotesize{\textsc{force}}} & \multicolumn{2}{c|}{\scriptsize{\textsc{Sketcher}}}  & \multicolumn{2}{c|}{\footnotesize{\textsc{force}}}  & \multicolumn{2}{c|}{\scriptsize{\textsc{Sketcher}}} & \multicolumn{2}{c}{\footnotesize{\textsc{force}}} \\
     & \scriptsize{\textsc{Calls}} & \scriptsize{\textsc{Time}} & \scriptsize{\textsc{Calls}} & \scriptsize{\textsc{Time}} &
     \scriptsize{\textsc{Calls}} & \scriptsize{\textsc{Time}} & \scriptsize{\textsc{Calls}} & \scriptsize{\textsc{Time}} &
     \scriptsize{\textsc{Calls}} & \scriptsize{\textsc{Time}} & \scriptsize{\textsc{Calls}} & \scriptsize{\textsc{Time}} \\
     \midrule
\footnotesize{\texttt{\SIMPLE}} & \footnotesize{2}   & \footnotesize{0.319} & \footnotesize{8} & \footnotesize{0.648} & \footnotesize{2} & \footnotesize{0.351}  & \footnotesize{16} & \footnotesize{1.250} & \footnotesize{2}   & \footnotesize{0.373} & \footnotesize{256} & \footnotesize{19.24} \\
% \footnotesize{\texttt{heapsort.c}} & \scriptsize{\texttt{invgen}} & \footnotesize{2} & \footnotesize{36} & \footnotesize{60}  & \footnotesize{0.079}   & \footnotesize{2.5$\times$} & \footnotesize{0.532} & \footnotesize{1.2$\times$} & \footnotesize{0.402} & \footnotesize{2.3$\times$}  \\
  \footnotesize{\texttt{\Loop}} &  \footnotesize{4}   & \footnotesize{0.638} & \footnotesize{8} & \footnotesize{0.614} & \footnotesize{4} & \footnotesize{0.658}  & \footnotesize{16} & \footnotesize{1.228} & \footnotesize{4}   & \footnotesize{1.667} & \footnotesize{256} & \footnotesize{18.95}     \\ %\hhline{||-----------||}
\footnotesize{\texttt{\LoopCond}} &   \footnotesize{2}   & \footnotesize{0.392} & \footnotesize{8} & \footnotesize{0.639} & \footnotesize{2} & \footnotesize{0.448}  & \footnotesize{16} & \footnotesize{1.251} & \footnotesize{2}   & \footnotesize{0.778} & \footnotesize{256} & \footnotesize{19.64}       \\ %\hhline{||-----------||}
 \footnotesize{\texttt{\Welfare}} &  \footnotesize{4}   & \footnotesize{0.660} & \footnotesize{8} & \footnotesize{0.650} & \footnotesize{5} & \footnotesize{0.923}  & \footnotesize{16} & \footnotesize{1.205} & \footnotesize{10}   & \footnotesize{1.476} & \footnotesize{256} & \footnotesize{19.69}     \\ \footnotesize{\texttt{\Salesman}} &  \footnotesize{2}   & \footnotesize{0.406} & \footnotesize{8} & \footnotesize{0.689} & \footnotesize{2} & \footnotesize{0.417}  & \footnotesize{16} & \footnotesize{1.359} & \footnotesize{2}   & \footnotesize{0.424} & \footnotesize{256} & \footnotesize{19.41}     \\
\bottomrule
\end{tabular}
\end{table*}

\paragraph*{Implementation}
We have developed a prototype model synthesizer, %in \textsc{Java},
called \SPLSketch, for resolving \Promela\, sketches.
It uses the ANTLR parser generator \cite{ANTLR}
for processing \Promela\, code, while projections and
abstractions of \texttt{\#if}-enriched \Promela\, code are implemented using source-to-source
transformations. % as described in Section~\ref{sec:sketches}.
It calls the \SPIN\, \cite{spin} to verify the generated \Promela\, models.
%In case of negative results, the tool analyzes the reported counterexamples to do refinement.
If a counterexample trace is returned, the tool inspects the trace by using \SPIN's
simulation mode, and generates refined abstractions.
Our tool is written in \textsc{Java} and consists of around 2K LOC.

\paragraph*{Experiment setup and Benchmarks}
All experiments are executed on a 64-bit Intel$^\circledR$Core$^{TM}$
i5 CPU, Lubuntu VM, with 8 GB memory.
%All times are reported as average over five independent executions.
The implementation, benchmarks, and all results are available from: \url{https://github.com/aleksdimovski/Promela_sketcher}.
%In our experiments, we use our tool based on lifted analysis via decision trees, denoted by \SPLSketch,
%where Polyhedra domain is used for properties in leaf and decision nodes.
We compare our approach with the \textsc{Brute force} enumeration approach
  that generates all possible sketch realizations and verifies them using \SPIN\, one by one.
For each experiment, we report: \textsc{Time} which is the total
time to resolve a sketch in seconds; and \textsc{Calls} which is the number of
times \SPIN\, is called. We show performances for three different sizes of holes:
3-, 4-, and 8-bits.
%For \SPLSketch, along with the total time it takes to complete we also
%report in parenthesis the time taken by \SPIN\, to perform model checking tasks.
We only measure the model checking \SPIN\, times
to generate a process analyser (pan) and to execute it.
We do not count the time for compiling pan, as it is due to a design decision in
\SPIN\, rather than its verification algorithm.
The evaluation is performed on several suitably adjusted
 \Promela\, sketches collected from the \Sketch\, project
 \cite{DBLP:journals/sttt/Solar-Lezama13}, SyGuSComp \cite{DBLP:conf/fmcad/AlurBJMRSSSTU13},
and \SPIN\, \cite{spin} (see benchmarks in App.~\ref{app:bench}).
%The chosen benchmarks are suitably adjusted to the sketching problem we consider.
%We use the following benchmarks:
%\SIMPLE\, (Fig.~\ref{fig:sketch1}), \Loop\, (Fig.~\ref{fig:bench1}),
%\LoopCond\, (Fig.~\ref{fig:bench2}), \Welfare\, (Fig.~\ref{fig:bench3}), and \Salesman\, (Fig.~\ref{fig:bench4}).

\paragraph*{Performance Results}
%We now present the performance results of our empirical study and discuss the implications.
Table~\ref{fig:performance1}  shows the results of synthesizing our benchmarks.
%Note that the \textsc{Brute force} enumeration approach %calls $2^n$ times \SPIN\, for each benchmark,
%where $n$ is the size of the holes in the sketch. Hence, \textsc{Brute force}
%calls \SPIN\, once for each possible hole realisation.

\SPLSketch\, needs two iterations and two calls to \SPIN\, to resolve the \SIMPLE\, sketch given in Fig.~\ref{fig:sketch1}
by reporting that the  hole \texttt{??} can be replaced with an integer value from $[Min,2]$.

Hence, it significantly outperforms the \textsc{Brute force} approach.

The \Loop\, sketch \cite{DBLP:journals/sttt/Solar-Lezama13} in Fig~\ref{fig:bench1}
contains one hole \texttt{??} represented by feature $A$.
The coarsest abstract model has an \texttt{if} statement with one optional sequence
`$\texttt{do} :: (\texttt{x>n}) \to \ldots \texttt{od}$' for each possible value \texttt{n} of feature $A$.
\SPLSketch\, reports counterexamples for the cases $(A\!=\!Min), \ldots, (A\!=\!4)$, and then
we obtain the correct solution for the abstraction $(5 \!\leq\! A \!\leq\! Max)$.
We have slightly changed \Loop\, by replacing \texttt{x:=10} with  \texttt{x:=??},
thus obtaining a sketch \Loop'\, with two holes represented by two features $A_1$ and $A_2$.
%The coarsest abstraction is $(Min_1 \!\leq\! A_1 \!\leq\! Max_1) \land (Min_2 \!\leq\! A_2 \!\leq\! Max_2)$.
All reported counterexamples have specific values for $A_1$ and $A_2$, which are
used to define refined abstract models.

\SPLSketch\, needs two iterations to resolve the \LoopCond\, sketch \cite{DBLP:journals/sttt/Solar-Lezama13} in Fig~\ref{fig:bench2}, where we use feature $A$ for the hole \texttt{??}.
In the first iteration, \SPIN\, reports an error trace %for the coarsest abstract model
that corresponds to the case when $(A=2)$.
In the next iteration, we obtain the correct answer for the abstraction $(A \leq 1)$.

The \Welfare\, sketch \cite{spin} in Fig~\ref{fig:bench3} is a problem due to Feijen. % described in \cite{DBLP:books/sp/Gries81}.
There are three ordered lists of integers \texttt{a}, \texttt{b}, and \texttt{c}.
At least one element appears in all three lists.
Find the smallest indices \texttt{i}, \texttt{j}, and \texttt{k}, such that \texttt{a[i]=b[j]=c[k]}.
That is, we want to find the first element that appears in all three lists.
The list \texttt{c} is initialized in such a way that concrete values assigned to the first $n-1$ elements
do not appear in lists \texttt{a} and \texttt{b}, and the last $n$-th element of \texttt{c} is assigned to the hole \texttt{??}.
Hence, the hole \texttt{??} should be replaced with the smallest value that appear also
in lists \texttt{a} and \texttt{b}.
\SPLSketch\, successfully partitions the configuration space and finds the
correct solutions for various values assigned to lists \texttt{a}, \texttt{b}, and \texttt{c}.
The number of iterations needed depends on the content of \texttt{a}, \texttt{b}, and \texttt{c}.

The \Salesman\, sketch \cite{spin} is a well-known optimisation problem, whose \Promela\, solution is given in Fig.~\ref{fig:bench4}.
Given a list of \texttt{N} cities and \texttt{distances} between each pair of cities,
it asks to find the shortest possible \texttt{tour} that visits each city and returns to
the origin city.
We now use our approach to find the shortest tour through the cities.
We initialize variable \texttt{MAX} to an integer hole \texttt{??}.
Whenever there exists a shorter tour than the one assigned to \texttt{MAX},
the given LTL property \texttt{p} fails and a counterexample is reported.
Therefore, the LTL property \texttt{p} will be correct only when \texttt{MAX} is
initialized to the value less or equal to the shortest possible tour.
\SPLSketch\, successfully finds this value for \texttt{??} in only two iterations.
In the first iteration, it reports a counterexample with a tour of length ($n+1$) that is greater than the shortest possible tour that is of length $n$.
Then, in the second iteration, the abstraction $(A \leq n)$ satisfies the property \texttt{p}.

We can see from Table~\ref{fig:performance1} that \SPLSketch\, significantly
outperforms \textsc{Brute force}. On our benchmarks,
it translates to speed ups that range from 1.2$\times$ to 3.5$\times$ for 4-bits holes,
and from 11.5$\times$ to 51.4$\times$ for 8-bits holes.
This is due to the fact that the number of calls to \SPIN\, and the number of
partitionings of \Kk\, that share the same counterexamples or correct traces in \SPLSketch\,
are much less than the configuration space \Kk\, that is inspected one by one using \SPIN\,
in the \textsc{Brute force}. % approach.

%% file: conclusion.tex
In this paper, we employ techniques from product-line lifted model checking for
automatically resolving of model sketches.
By means of an implementation and a number of experiments,
we confirm that our technique is effective and works well on a
variety of \Promela\, benchmarks and LTL properties.
\begin{comment}
Thanks to its simplicity, our solution of model sketching can be applied
in other instances of model synthesis.
For example, model repair \cite{DBLP:journals/corr/Chatzieleftheriou15} amounts to obtaining a repaired version
of an erroneous model. Each admissible repaired version can be viewed
as a family member, so model repair can be reduced to the lifted model checking problem.
%In the future, we would like to extend
%the model sketching approach to also support CTL properties.
%An integration of neural networks and model sketching is also possible for
%learning complete models. % from examples and natural language instructions.

Our approach reports several, and very often, all solutions to a sketching problem.
Very often among two correct solutions that satisfy a given specification, one may be preferred
to another. In the future, we plan to consider a synthesis task where the synthesized model
must %not only meet a boolean specification, but
 be optimal with respect to a
quantitative objective \cite{DBLP:conf/cav/BloemCHJ09}. 
\end{comment}

%% file: proofs.tex
Let $H$ be a set of holes in the sketch $\hat{P}$. We define a \emph{control function}
$\varphi : \Phi = H \to \Zz$ to describe the value of each hole in $\hat{P}$. %Thus, $\varphi$ fully describes a candidate solution to the sketch.
We write $\sbr{\hat{P}}^{\varphi}_{TS}$ for TS obtained by replacing holes in $\hat{P}$ according to $\varphi$.
\begin{theorem} \label{thm:rewrite}
Let $\hat{P}$ be a sketch and $\varphi$ be a control function, s.t.
features $A_1, \ldots, A_n$ correspond to holes $\texttt{??}_1, \ldots, \texttt{??}_n$. % from $H$.
We define a configuration $k \in \Kk$, s.t. $k(A_i)=\varphi(\texttt{??}_i)$ for $1 \!\leq\! i \leq n$.
Let $\overline{P} = \texttt{Rewrite}(\hat{P})$. We have: $\sbr{\hat{P}}^{\varphi}_{TS} \equiv \sbr{\pi_k(\sbr{\overline{P}}_{FTS})}_{TS}$.
\end{theorem}
\begin{proof}
By induction on the structure of $\hat{P}$.
The only interesting case is a basic statement $stm[\texttt{??}_i]$ for rule (\ref{transf.1}), since in all
other cases we have identity transformations.
\[
\begin{array}{@{}l@{}}
\Scale[0.86]{\sbr{\pi_{k}(\sbr{\texttt{\#if}\!::\! \texttt{(A}_i=n) \rightarrow stm[n] \ldots \!::\!\texttt{(A}_i=n') \rightarrow stm[n'] \, \texttt{\#endif}
}_{FTS})}_{TS} } \\
\Scale[0.86]{  \overset{\text{def. of $\pi_k$}} = \sbr{stm[k(\texttt{A}_i)]}_{TS} \overset{\text{hypoth.}}= \sbr{stm[\varphi(\texttt{??}_i)]}_{TS} \overset{\text{def. of \texttt{??}}}= \sbr{stm[\texttt{??}_i]}^{\varphi}_{TS} }
\end{array}
\]
\end{proof}

\begin{theorem} \label{theo:transform}
Let $\overline{P}$ be a \Promela\, family and $\sbr{\overline{P}}_{FTS}$ be its FTS.
Then: $\pi_{\sbr{\psi}}(\sbr{\overline{P}}_{FTS}) \!\equiv\! \sbr{\pi_{\sbr{\psi}}(\overline{P})}_{FTS}$ and
$\joinasym_{\Kk}(\sbr{\overline{P}}_{FTS}) \!\equiv\! \sbr{\joinasym_{\Kk}(\overline{P})}_{TS}$.
\end{theorem}
\begin{proof}
By induction on the structure of $\overline{P}$.
The only interesting case is ``$\texttt{\#if}$'',
since in all other cases we have an identity translation.
We can see that projection $\pi_{\sbr{\psi}}$ and abstraction $\joinasym_{\Kk}$ are applied on feature
expressions $\psi$, which can be introduced in FTSs only through ``$\texttt{\#if}$''-s.
Thus, it is the same whether $\pi_{\sbr{\psi}}$ and $\joinasym_{\Kk}$ are applied directly on FTS $\sbr{\overline{P}}_{FTS}$ after the FTS is built
by following the operational semantics of ``$\texttt{\#if}$'', or $\pi_{\sbr{\psi}}$ and $\joinasym_{\Kk}$ are first applied
on ``$\texttt{\#if}$''-s using rules (\ref{transf.2}), (\ref{transf.3}) and then FTS is built.
\end{proof}

\begin{theorem} \label{theo:synth}
\texttt{SYNTHESIZE}$(\hat{P},\phi$) is correct and terminates.
\end{theorem}
\begin{proof}
The procedure
\texttt{SYNTHESIZE($\hat{P},\phi$)} terminates since all steps in it are terminating.
The correctness of \texttt{SYNTHESIZE($\hat{P},\phi$)} follows from the correctness of \texttt{Rewrite} (see Theorem~\ref{thm:rewrite}),
\texttt{ARP} and syntactic transformations (see Theorem~\ref{theo:transform}).
\end{proof}

%% file: bench.tex
\section{Benchmarks} \label{app:bench}

\begin{figure*}[h]
\centering
\begin{minipage}[b]{0.4\textwidth}
\centering
	$\begin{array}{ll}
\hline
 & \texttt{init} \ \{ \\
 \quad & \quad \texttt{byte} \ \texttt{x}; \, \texttt{int} \ \texttt{y}; \\
 \quad & \quad \texttt{do} :: \texttt{break} \, :: \texttt{x++} \ \texttt{od}; \\
\quad & \quad \impassign{\texttt{y}}{\texttt{??*x};} \\
\quad & \quad \texttt{assert} \, (\texttt{y} \leq \texttt{x+x}) \, \} \\
\hline
	\end{array}$
\vspace{-1mm}
\caption{\SIMPLE\, sketch.}
\label{fig:sketch1}
\end{minipage}
\begin{minipage}[b]{0.45\linewidth}
	$\begin{array}{@{} l @{}}
\hline
  \texttt{init} \ \{ \\
  \quad  \texttt{byte} \ \texttt{x:=10}; \\
  \quad  \texttt{int} \ \texttt{y:=0}; \\
  \quad  \texttt{do} :: (\texttt{x>??}) \to \texttt{x--}; \\
  \quad \qquad \qquad \qquad \ \texttt{y++} \\
  \quad \quad \ :: \texttt{else} \to \texttt{break} \\
  \quad  \texttt{od}; \\
 \quad  \texttt{assert} \, (\texttt{y} < \texttt{6}) \, \} \\
\hline
	\end{array}$
\vspace{-1mm}
\caption{\Loop\, sketch.}
\label{fig:bench1}
\end{minipage}
\\
\begin{minipage}[b]{0.4\linewidth}
$\begin{array}{@{} l @{}}
\hline
  \texttt{init} \ \{ \\
  \quad  \texttt{byte} \ \texttt{x}; \texttt{int} \ \texttt{y:=0}; \\
  \quad  \texttt{do} :: \texttt{break} \, :: \texttt{x++} \ \texttt{od}; \\
  \quad  \texttt{do} :: (\texttt{x>0}) \to \texttt{x--};  \\
  \quad \quad \qquad \texttt{if} :: (\texttt{y<??}) \to \texttt{y++} \\
  \quad \qquad \qquad \ :: \texttt{else} \to \texttt{y--} \ \texttt{fi}; \\
  \quad \quad \ :: \texttt{else} \to \texttt{break} \  \texttt{od}; \\
  \quad  \texttt{assert} \, (\texttt{y} \leq \texttt{1}) \, \} \\
\hline
\end{array}$
\vspace{-1mm}
\caption{\LoopCond\, sketch.}
\label{fig:bench2}
\end{minipage}
\begin{minipage}[b]{0.4\linewidth}
$\begin{array}{@{} l @{}}
\hline
  \texttt{int} \ \texttt{a[5], b[5], c[5]}; \\
  \texttt{init} \ \{ \\
  \quad  \texttt{byte} \ \texttt{i, j, k}; \\
  \quad  \texttt{a[0]:=1}; \ldots \ \texttt{a[4]:=18}; \\
  \quad  \texttt{b[0]:=4}; \ldots \ \texttt{a[4]:=25}; \\
  \quad  \texttt{c[0]:=5}; \ldots \ \texttt{c[4]:=??}; \\
  \quad  \texttt{do} :: \texttt{a[i]<b[j]} \land \texttt{i<4} \to \texttt{i++};  \\
  \quad \quad \ :: \texttt{b[j]<c[k]} \land \texttt{j<4} \to \texttt{j++};  \\
  \quad \quad \ :: \texttt{c[k]<a[i]} \land \texttt{k<4} \to \texttt{k++};  \\
  \quad \quad \ :: \texttt{else} \to \texttt{break} \  \texttt{od}; \\
 \quad  \texttt{assert} (\texttt{a[i]=b[j]} \!\land\! \texttt{b[j]=c[k]} \\
 \quad \qquad \qquad \land \texttt{c[k]=a[i]}) \, \} \\
\hline
\end{array}$
\vspace{-1mm}
\caption{\Welfare\, sketch.}
\label{fig:bench3}
\end{minipage}
\end{figure*}

\begin{figure*}
\centering
\begin{minipage}[b]{0.49\linewidth}
	$\begin{array}{@{} l @{}}
\hline
    \texttt{byte} \ \texttt{N:=4, MAX, distance[16]}; \\
    \texttt{byte} \ \texttt{city, dest, tour, seen}; \\
    \texttt{bool} \ \texttt{visited[4]}; \\
    \texttt{\#define} \ \texttt{Dist(a,b)} \ \texttt{distance[4*a+b]} \\
    \texttt{inline travel2(dest)} \ \{ \\
    \quad  (\texttt{city != dest} \land \texttt{tour} \leq \texttt{MAX}) \to  \\
    \quad  \texttt{tour := tour + Dist(city,dest)}  \\
    \quad  \texttt{city := dest}  \\
    \quad  \texttt{if} :: (\neg \texttt{visited[city]}) \to \\
    \quad \qquad \qquad \qquad \ \texttt{visited[city]:=true}; \quad \ \\
    \quad \qquad \qquad \qquad \ \texttt{seen++} \\
    \quad \quad \ :: \texttt{else} \to \texttt{break} \ \  \texttt{fi}  \} \\
\hline
	\end{array}$
\vspace{-1mm}
\end{minipage}
\begin{minipage}[b]{0.49\linewidth}
$\begin{array}{@{} l @{}}
\hline
  \texttt{init} \ \{ \\
  \quad  \texttt{MAX:=??}; \\
  \quad  \texttt{Dist(0,1)=20}; \ldots \ \texttt{Dist(3,2)=12}; \\
  \quad  \texttt{do} :: \texttt{select (dest: 0 .. (N-1))} \to \\
  \quad \qquad \qquad \qquad \qquad \texttt{travel2(dest)}  \\
  \quad  \texttt{od}; \\
  \quad  \texttt{ltl p} \ \{ [\,] (\texttt{seen<N} \land \texttt{tour>MAX}) \, \} \\
  \} \\
\hline
\end{array}$
\vspace{-1mm}
\end{minipage}
\caption{\Salesman\, sketch.}
\label{fig:bench4}
\end{figure*}